\documentclass{article}
\usepackage{fullpage}
\usepackage{amsmath,amssymb,epsfig}

\newcommand{\xv}{\mathbf{x}}
\newcommand{\yv}{\mathbf{y}}

\newcommand{\Dc}{\mathcal{D}}
\newcommand{\Gc}{\mathcal{G}}
\newcommand{\Cc}{\mathcal{C}}
\newcommand{\Rc}{\mathcal{R}}
\newcommand{\Sc}{\mathcal{S}}
\newcommand{\Tc}{\mathcal{T}}
\newcommand{\Vc}{\mathcal{V}}

\newcommand{\FF}{\mathbb{F}}

\newtheorem{example}{Example}
\newtheorem{lemma}{Lemma}

\newtheorem{prop}{Proposition}

\newtheorem{definition}{Definition}

\newenvironment{proof}{
\noindent
{\bf Proof.}
}{
\begin{flushright}$\blacksquare$\end{flushright}
}

\begin{document}
\title{An Authentication Code against Pollution Attacks \\in Network Coding}
\author{
Fr\'ed\'erique Oggier
and Hanane Fathi
\thanks{
F. Oggier is with Division of Mathematical Sciences, School of Physical
and Mathematical Sciences, Nanyang Technological University, Singapore.
Email:frederique@ntu.edu.sg. H. Fathi is with Center for
TeleInfrastuktur, Aalborg University, Denmark. Email:
hf@es.aau.dk. Part of this work was presented in an invited paper
at Allerton conference 2008. }} \maketitle

\begin{abstract}
Systems exploiting network coding to increase their throughput suffer
greatly from pollution attacks which consist of injecting malicious packets
in the network. The pollution attacks are amplified by the network coding
process, resulting in a greater damage than under traditional routing. In
this paper, we address this issue by designing an unconditionally secure
authentication code suitable for multicast network coding. The proposed
scheme is robust against pollution attacks from outsiders, as well as
coalitions of malicious insiders. Intermediate nodes can verify the integrity
and origin of the packets received without having to decode, and thus detect
and discard the malicious messages in-transit that fail the verification.
This way, the pollution is canceled out before reaching the destinations.
We analyze the performance of the scheme in terms of both multicast
throughput and goodput, and show the goodput gains. We also discuss
applications to file distribution.
\end{abstract}

%
%

\section{Introduction}

Network coding was first introduced in \cite{alsch} as an innovative approach
to characterize the rate region of multicast networks. Network
coding allows intermediate nodes between the source(s) and the destinations not
only to store and forward, but also to encode the received packets before
forwarding them. In \cite{li99}, Li et.~al showed that linear coding suffices 
to achieve the max-flow from the source to each receiving node in multicast 
networks, where intermediate nodes generate outgoing packets as linear 
combinations of their incoming packets.
In line with \cite{li99}, \cite{koetter01} gave an algebraic framework for
linear network coding with further developments for arbitrary networks
and robust networking. For practical issues, \cite{chou07} proposed a network
coding framework that allows to deal with random packet loss, change of
topology and delays.

Network coding offers various advantages not only for maximizing
the usage of network resources but also for robustness to network
impairments and packet losses. Various applications of network
coding have therefore appeared ranging from file download and
content distribution in peer-to-peer networks
\cite{Avalanche,GR05,GMR05} to distributed file storage systems
\cite{C08,DGWR07}.

While much of the literature on network coding discusses network capacity
or throughput, it is also natural to wonder about the impact of network
coding on network security. Pollution attacks, which consist of injecting
malicious packets in the network, are for example more dangerous for the
systems exploiting network coding than for those using traditional routing.
Indeed, in this scenario, malicious packets may come from the modification
of received packets by a malicious intermediate node, or from the creation
of bogus packets then injected in the network by an outside adversary. 
With no integrity check performed for packets in transit in the network, 
an honest intermediate node receiving a single malicious packet would perform 
the encoding of the malicious packet with other packets resulting in multiple
corrupted outgoing packets that are then forwarded on to the next
nodes. The corrupted packets propagate then all through the
network which creates severe damages, amplified by the network coding process.

\subsection{Authentication techniques}

One way to address the pollution attack problem is through authentication
techniques. Packets in transit at the intermediate nodes should be 
authenticated before being encoded and forwarded, to verify both their origin 
and their content.
The goal is to achieve authentication even in presence of both inside and 
outside attackers who can observe the messages flowing through the network 
and inject selected messages. The success of their attacks depends on their 
ability in sending a message that will be accepted as valid (i.e., 
{\em impersonation attack}) or in observing a message and then altering 
the message content (i.e. {\em substitution attack}) in such a way that 
intermediate nodes and destinations cannot detect it.

Let us recall that authentication consists of the following properties, 
though we will focus here only on the first two:
\begin{itemize}
\item
{\em data integrity}: protecting the data from any modification by malicious
entities,
\item
{\em data origin authentication}: validating the identity of the origin of
the data,
\item
{\em non-repudiation}: guaranteeing that the origin of the data cannot
deny having created and sent data.
\end{itemize}

To satisfy these properties, messages at the source are appended 
either a digital signature, a message authentication code (MAC) or
an authentication code (also called tag). There exist subtle
differences among these techniques. First, MAC and authentication
codes ensure data integrity and data origin authentication while
digital signatures provide also non-repudiation. Second, MACs,
authentication codes, and digital signatures should be
differentiated depending on what type of security they achieve:
computational security (i.e., vulnerable against an attacker that
has unlimited computational resources) or unconditional security
(i.e., robust against an attacker that has unlimited computational
resources). MACs are proven to be computationally secure while the
security of authentication codes is unconditional \cite{stinson}.
Digital signature schemes exist for both computational security
and unconditional security. However while computationally secure
digital signatures can be verified by anyone with a public
verification algorithm, the unconditionally secure digital
signatures can only be verified by intended receivers as it is for
MACs and authentication codes \cite{Hanaoka00}.


\subsection{Related work}
\label{subsec:relwork}

Several authentication schemes have been recently proposed in the
literature to detect polluted packets at intermediate nodes
\cite{GR06,charlesjain,Yu, ZhaoMed07,BFKW08}. All of them are
based on cryptographic functions with computational assumptions,
as detailed below.

The scheme in \cite{GR06} for network-coded content distribution allows
intermediate nodes to detect malicious packets injected in the network and to
alert neighboring nodes when a malicious packet is detected. It uses a
homomorphic hash function to generate hash values of the encoded blocks of
data that are then sent to the intermediate nodes and destinations prior to
the encoded data. The transmission of these hash values is performed over a
pre-established secure channel which makes the scheme impractical. The use of
hash functions makes the scheme fall into the category of computationally
secure schemes.

The signature scheme in \cite{charlesjain} is a homomorphic signature scheme
based on Weil pairing over elliptic curves, while the one proposed in
\cite{Yu} is a homomorphic signature scheme based on RSA. For both schemes,
intermediate nodes can authenticate the packets in transit without decoding,
and generate a verifiable signature of the packet that they have just encoded
without knowing the signer's secret key. However, these schemes require one
key pair for each file to be verified, which is not practical either.

The signature scheme proposed in \cite{ZhaoMed07} uses a standard signature
scheme based on the hardness of the discrete logarithm problem. The blocks of
data are considered as vectors spanning a subspace. The signature is not
performed on vectors containing data blocks, but on vectors orthogonal to all
data vectors in the given subspace. The signature verification allows to check
if the received vector belongs to the data subspace. The security of their
scheme holds in that no adversary knowing a signature on a given subspace of
data vectors is able to forge a valid signature for any vector not in this
given subspace. This scheme requires also fresh keys for every file.
 
Finally, the signature schemes
given in \cite{BFKW08} follow the approach given in \cite{ZhaoMed07} with
improvements in terms of public key size and per-packet overhead. The signature
schemes proposed are designed to authenticate a linear subspace formed by the
vectors containing data blocks. Signatures on a linear subspace are sufficient
to authenticate all the vectors in this same subspace. With these schemes, a
single public key can be used to verify multiple files.


\subsection{Organization and contribution}

In this paper, we propose an unconditionally secure solution that
provides multicast network coding with robustness against pollution attacks.
Our solution allows intermediate nodes and destinations to verify the data 
origin and integrity of the messages received without decoding, and 
thus to detect and discard the malicious messages that fail the verification.
It is important to note that destinations must receive a sufficient number
of uncorrupted messages to decode and recover the entire file sent by the
source. However, our solution provides the destinations with the ability
to filter out corrupted messages and to have them filtered out by
intermediate nodes as well.

Our scheme here aims for unconditional security. We rely on
information theoretic strength rather than on problems that are thought to be
hard as in \cite{GR06,Yu, ZhaoMed07,BFKW08}. Unconditional authentication 
codes have led to the development of multi-receiver authentication codes 
\cite{desmedt, safavi98} that are highly relevant in the context of network 
coding. Multi-receiver authentication codes allow any one of the receivers (in
the context of network coding, that may be intermediate nodes and destinations) 
to verify the integrity and origin of a received message but require the 
source to be designated. Our scheme is inspired from the $(k,V)$ 
multi-receiver authentication code proposed in \cite{safavi98} that is robust 
against a coalition of $k-1$ malicious receivers amongst $V$ and in which 
every key can be used to authenticate up to $M$ messages. We define and adapt 
the use of $(k,V)$ multi-receiver authentication codes to network coding so 
that intermediate nodes can detect malicious packets without having to decode.

Our scheme is adaptive to the specifications of the application in use and the 
network setting. Its efficiency is scenario-dependent. The communication and 
computational costs are function of parameters related to the application in 
use (i.e., the number $M$ of messages to be authenticated under the same key 
and the length $l$ of the messages) and to the network setting (i.e., the 
number of colluded malicious adversaries $k-1$ to be considered). However for 
the communication cost, one independent advantage exists over the previous 
schemes. Our scheme is particularly efficient in terms of communication 
overhead, since contrarily to all existing schemes 
\cite{GR06,Yu, ZhaoMed07,BFKW08}, it requires one single symbol only for 
tracking purposes.

We give a multicast goodput analysis to assess the impact of pollution
attacks on multicast throughput and to show how much goodput gain our
scheme offers. We show how our scheme can be used for applications such as
content and file distribution.

The rest of the paper is organized as follows. In Section \ref{sec:netcod}, we
briefly present the network coding model we consider and define what are
authentication codes in general and in particular for network coding. Section
\ref{sec:sign} presents the authentication scheme, whose analysis is presented
both in Section \ref{sec:analysis} for security, and in Section \ref{sec:perf}
for performance. Section \ref{sec:applic} shows how our scheme can be used
for content and file distribution. Future work is addressed in the conclusion.

%
%

\section{A Network Coding Setting for Authentication Codes}
\label{sec:netcod}

We start by introducing the multicast network coding model we are 
considering. Since we are not aware of prior work on authentication codes 
for network coding, we then propose a definition of authentication codes for 
multicast network coding.

\subsection{The multicast network coding model}

The model of network we consider is an acyclic graph having unit capacity
edges, with a single source $S$, which wants to send a set of $n$ messages to
$T$ destinations $D_1,\ldots,D_T$. Messages are seen as sequences of elements
of a finite field with $q$ elements, denoted by $\FF_q$. Each edge $e$ of
the graph carries a symbol $y(e) \in \FF_q$ at a time. For a node of the graph, 
the symbols on its outgoing edges are linear combinations, called {\em local
encoding}, of the symbols entering the node through its incoming edges.
If $x_1,\ldots,x_n$ are the symbols to be sent by the source $S$ at a time, 
we have by induction that on any edge $e$, $y(e)$ is actually a linear 
combination of the source symbols, that is $y(e)=\sum_{i=1}^ng_i(e)x_i$, 
where the coefficients $g_i(e)$ describe the coding operation. The vector 
$g(e)=[g_1(e),\ldots,g_n(e)]$ is thus called the {\em global encoding vector} 
along the edge $e$. We can describe the messages received by a node in the 
network with $h$ incoming edges $e_1,\ldots,e_h$ by the following matrix 
equation:
\begin{eqnarray*}
\left(
\begin{array}{c}
y(e_1)\\
\vdots\\
y(e_h)
\end{array}
\right)
&=&
\underbrace{
\left(
\begin{array}{ccc}
g_1(e_1) & \ldots & g_n(e_1) \\
\vdots & & \vdots \\
g_1(e_h) &\ldots & g_n(e_h) \\
\end{array}
\right)}_G
\left(
\begin{array}{c}
x_1\\
\vdots\\
x_n
\end{array}
\right)
\in \FF_q^n
\end{eqnarray*}
where $G$ is called a {\em transfer matrix}. In particular, the destination
nodes $D_i$, $i=1,\ldots,T$, can recover the source symbols $x_1,\ldots,x_n$,
assuming that their respective transfer matrix $G_{D_i}$ has rank $n$,
$i=1,\ldots,T$ (this also means $h\geq n$). In this paper, we are not concerned
about the existence of global encoding vectors, and we thus assume that we deal
with a network for which suitable linear encoding vectors exist, so that destination
nodes are able to decode the received packets correctly.

We can packetize the symbols $y(e)$ flowing on each edge $e$ into vectors
$\yv(e)=[y_1(e),\ldots,y_N(e)] \in \FF_q^N$, and likewise, the source symbols
$x_i$ can be grouped as $\xv_i=[x_{i,1},\ldots,x_{i,N}] \in \FF_q^N$, so that
the equation at a node with $h$ incoming edges can be rewritten as
\begin{equation}\label{eq:v_i}
\left(
\begin{array}{c}
\yv(e_1)\\
\vdots\\
\yv(e_h)
\end{array}
\right)
=
\left(
\begin{array}{ccc}
g_1(e_1) & \ldots & g_n(e_1) \\
\vdots & & \vdots \\
g_1(e_h) & \ldots & g_n(e_h) \\
\end{array}
\right)
\left(
\begin{array}{c}
\xv_1\\
\vdots\\
\xv_n
\end{array}
\right)\in\FF_q^{h\times N}
\end{equation}
or equivalently
\[
\left(
\begin{array}{c}
\yv(e_1)\\
\vdots\\
\yv(e_h)
\end{array}
\right)
=
G
\left(
\begin{array}{cccc}
x_{1,1} & x_{1,2} & \ldots & x_{1,N} \\
\vdots & \vdots & & \vdots\\
x_{n,1} & x_{n,2} & \ldots & x_{n,N}
\end{array}
\right)\in\FF_q^{h\times N}
\]
where $\xv_1,\ldots,\xv_n$ are the $n$ messages of length $N$ to be sent by
the source.

\begin{example}\label{ex:net1}\rm
Consider the small network (taken from \cite{koetter01}) as shown in
Fig.~\ref{fig:ex1}, where the source $S$ wants to send $n=3$ messages
$\xv_1,\xv_2,\xv_3 \in \FF_2^N$ to $T=1$ destination $D_1$, through two nodes
$R_1$ and $R_2$.

\begin{figure}
\begin{center}
\leavevmode
\epsfxsize=8cm
\epsfbox{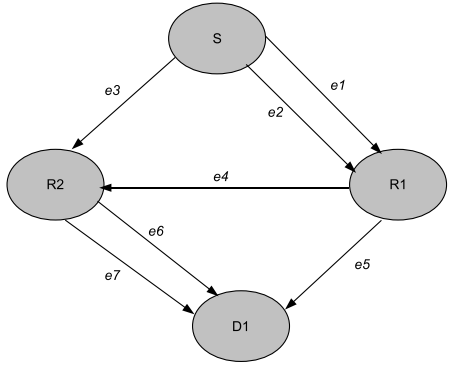}
\caption{A small example of network with one source $S$, one destination 
$D_1$ and two relay nodes $R_1$ and $R_2$. Global encoding vectors have  
coefficients in $\FF_2$.}
\label{fig:ex1}
\end{center}
\end{figure}

The source computes the vector
\[
\left(
\begin{array}{c}
\yv(e_1)\\
\yv(e_2)\\
\yv(e_3)\\
\end{array}
\right)=
\left(
\begin{array}{ccc}
g_1(e_1)&g_2(e_1)&g_3(e_1)\\
g_1(e_2)&g_2(e_2)&g_3(e_2)\\
g_1(e_3)&g_2(e_3)&g_3(e_3)\\
\end{array}
\right)
\left(
\begin{array}{c}
\xv_1 \\
\xv_2 \\
\xv_3
\end{array}
\right)
\]
as a linear combination of its three messages $\xv_1,\xv_2,\xv_3$ and
sends each $\yv(e_i)$  over the edge $e_i$, $i=1,2,3$.
The node $R_1$ receives $\yv(e_1)$ and $\yv(e_2)$, which it encodes
as follows using its global encoding vectors 
$g(e_4)=(\alpha_{11},\alpha_{12})$ and $g(e_5)=(\alpha_{21},\alpha_{22})$:
\begin{eqnarray*}
\left(
\begin{array}{c}
\yv(e_4)\\
\yv(e_5)
\end{array}
\right)
&=&
\left(
\begin{array}{cc}
\alpha_{11} & \alpha_{12} \\
\alpha_{21} & \alpha_{22} \\
\end{array}
\right)
\left(
\begin{array}{c}
\yv(e_1)\\
\yv(e_2)\\
\end{array}
\right)\\
&=&
\left(
\begin{array}{cc}
\alpha_{11} & \alpha_{12} \\
\alpha_{21} & \alpha_{22} \\
\end{array}
\right)
\left(
\begin{array}{ccc}
g_1(e_1)&g_2(e_1)&g_3(e_1)\\
g_1(e_2)&g_2(e_2)&g_3(e_2)\\
\end{array}
\right)
\left(
\begin{array}{c}
\xv_1 \\
\xv_2 \\
\xv_3
\end{array}
\right)\\
&=:&
\left(
\begin{array}{ccc}
g_1(e_4)&g_2(e_4)&g_3(e_4)\\
g_1(e_5)&g_2(e_5)&g_3(e_5)\\
\end{array}
\right)
\left(
\begin{array}{c}
\xv_1\\
\xv_2\\
\xv_3
\end{array}
\right)
\end{eqnarray*}
while the node $R_2$ gets
\begin{eqnarray*}
\left(
\begin{array}{c}
\yv(e_3)\\
\yv(e_4)
\end{array}
\right)
&=&
\left(
\begin{array}{ccc}
0 & 0 & 1 \\
\alpha_{11} & \alpha_{12} & 0 \\
\end{array}
\right)
\left(
\begin{array}{c}
\yv(e_1)\\
\yv(e_2)\\
\yv(e_3)
\end{array}
\right)
\\
&=&
\left(
\begin{array}{ccc}
0 & 0 & 1 \\
\alpha_{11} & \alpha_{12} & 0 \\
\end{array}
\right)
\left(
\begin{array}{ccc}
g_1(e_1) & g_2(e_1) & g_3(e_1) \\
g_1(e_2) & g_2(e_2) & g_3(e_2) \\
g_1(e_3) & g_2(e_3) & g_3(e_3)\\
\end{array}
\right)
\left(
\begin{array}{c}
\xv_1 \\
\xv_2 \\
\xv_3
\end{array}
\right)
\\
&=:&
\left(
\begin{array}{ccc}
g_1(e_3)&g_2(e_3) & g_3(e_3)\\
g_1(e_4)&g_2(e_4) & g_3(e_4)
\end{array}
\right)
\left(
\begin{array}{c}
\xv_1 \\
\xv_2 \\
\xv_3
\end{array}
\right).
\end{eqnarray*}
Denote by $g(e_6)=(\beta_{11},\beta_{12})$ and $g(e_7)=(\beta_{21},\beta_{22})$ 
the global encoding vectors of node $R_2$ corresponding to the edges 
$e_6$ and $e_7$ respectively. Finally, the destination gets
\begin{eqnarray*}
&\!\!\!\!\!&\left(
\begin{array}{c}
\yv(e_5)\\
\yv(e_6)\\
\yv(e_7)
\end{array}
\right)\\
&\!\!\!\!\!=& \!\!\!\!\!
\left(
\begin{array}{cccc}
\alpha_{21} & \alpha_{22} & 0 & 0 \\
0 & 0 & \beta_{11} & \beta_{12} \\
0 & 0 & \beta_{21} & \beta_{22}
\end{array}
\right)
\left(
\begin{array}{c}
\yv(e_1)\\
\yv(e_2)\\
\yv(e_3)\\
\yv(e_4)
\end{array}
\right)
\\
&\!\!\!\!\!=& \!\!\!\!\!
\left(
\begin{array}{cccc}
\alpha_{21} & \alpha_{22} & 0 & 0 \\
0 & 0 & \beta_{11} & \beta_{12} \\
0 & 0 & \beta_{21} & \beta_{22}
\end{array}
\right)\!\!\!
\left(
\begin{array}{ccc}
g_1(e_1) & g_2(e_1) & g_3(e_1) \\
g_1(e_2) & g_2(e_2) & g_3(e_2) \\
g_1(e_3) & g_2(e_3) & g_3(e_3)\\
\alpha_{11}g_1(e_1)+\alpha_{12}g_1(e_2) & \alpha_{11}g_2(e_1)+\alpha_{12}g_2(e_2)
& \alpha_{11}g_3(e_1)+\alpha_{12}g_3(e_2)\\
\end{array}
\right)\!\!\!
\left(
\begin{array}{c}
\xv_1 \\
\xv_2 \\
\xv_3
\end{array}
\right)
\\
&\!\!\!\!\!=:& \!\!\!\!\!
\underbrace{
\left(
\begin{array}{cccc}
g_1(e_5) & g_2(e_5) & g_3(e_5) \\
g_1(e_6) & g_2(e_6) & g_3(e_6)\\
g_1(e_7) & g_2(e_7) & g_3(e_7)\\
\end{array}
\right)}_{G_{D_1}}
\left(
\begin{array}{c}
\xv_1\\
\xv_2\\
\xv_3\\
\end{array}
\right).
\end{eqnarray*}
The destination $D_1$ can decode if the global vectors have been chosen
such that the transfer matrix $G$ is invertible. The global vectors are linear
combinations of the local encoding coefficients $\alpha_{ij}$ at $R_1$ and
$\beta_{ij}$ at $R_2$, $i,j=1,2$. There are many configurations
over $\FF_2$ such that $G$ is invertible. Take for example $g_1(e_1)= g_2(e_2)
=g_3(e_3)=1$, $g_1(e_2)=g_2(e_1)=g_3(e_1)=g_1(e_3)=g_2(e_3)=g_3(e_2)=0$, with
$\beta_{12}=\beta_{21}=\alpha_{12}=\alpha_{21}=1$ and $\beta_{11}=\beta_{22}=
\alpha_{11}=\alpha_{12}=0$ (this yields the transfer matrix
to be equal to the identity matrix).
\end{example}


\subsection{Authentication codes for network coding}

Since we are not aware of prior work on network coding authentication codes, 
let us start by recalling the setting for classical authentication schemes, 
as proposed by Desmedt et al. In \cite{desmedt}, the authors proposed a model 
for unconditionally secure authentication where one transmitter communicates
to multiple receivers who can not all be trusted. In this scenario, the
transmitter first appends a tag to a common message which is then broadcasted
to all the receivers, who can separately verify the authenticity of the tagged 
message using their own private secret key. There is among the receivers a 
group of malicious receivers, who use their secret key and all the previous 
messages to construct fake messages. A $(k,V)$ multi-receiver authentication 
system refers to a scheme where $V$ receivers are present, among which at 
most $k-1$ can cheat.
The malicious nodes can perform either an {\em impersonation attack}, if 
they try to construct a valid tagged message without having seen any 
transmitted message before, or a {\em substitution attack}, if they first 
listen to at least one tagged message before trying to fake a tag in such 
a way that the receiver will accept the tagged message. 
{\em Perfect protection} is obtained if the best chance of success in the 
attack is $1/|\Tc|$ where $|\Tc|$ is the size of tag space, namely, the 
attacker cannot do better than make a guess, and pick randomly one tag.

In \cite{safavi98}, the scheme of Desmedt et al. has been generalized to
the case where the same key can be used to authenticate up to $M$ messages.

The network coding scenario that we consider in this paper is a multicast
setting, where one source wants to send a set of messages to $T$ destinations.
In order to propose a definition of network coding authentication scheme, let
us first understand the main differences with respect to the classical 
multi-receiver scenario:
\begin{enumerate}
\item
The source does not broadcast the same message on all its outgoing edges,
but sends different linear combinations of the $n$ messages $\xv_1,\ldots,
\xv_n$, which means that the key used by the source to sign the
messages will be used more than once, actually at least as many times as there
are outgoing edges from the source.
\item
We are interested in a more general network scenario, where intermediate nodes 
play a role. In particular, it is relevant in the context of pollution attacks 
that not only destination nodes but also intermediate nodes may check the 
authenticity of the packets. We call such nodes in the network {\em verifying} 
nodes. This set may include part or all of the destination nodes 
$D_1,\ldots,D_T$. This makes a big difference in network coding, since while 
the destination nodes do have a transfer matrix to recover the message sent, 
this is not the case of regular intermediate nodes, which must perform the 
authentication check without being able a priori to decode.
\end{enumerate}

Based on the above considerations, we propose the following definition for 
multicast network coding.

\begin{definition}
We call a $(k,V,M)$ {\bf\em network coding authentication code} an 
authentication code for $V$ verifying nodes, which is unconditionally secure 
against either substitution or impersonation attacks done by a group of at most 
$k-1$ adversaries, possibly belonging to the verifying nodes, where the 
source can use the same key at most $M$ times.
\end{definition}

%
%

\section{The Authentication Scheme}
\label{sec:sign}

Recall that we have a single source $S$, which wants to multicast $n$ 
messages to $T$ destinations $D_1,\ldots,D_T$. We will denote the set of 
messages by  $s_1,\ldots,s_n$ to refer to the actual data to be sent, while 
we keep the notation $\xv_1,\ldots,\xv_n \in \FF_q^N$ for the whole packets, 
including the authentication tag. Each message $s_i$ is of length $l$, 
$s_i=(s_{i,1},\ldots,s_{i,l})$, so that while each symbol $s_{i,j}$ belongs 
to $\FF_q$, we can see the whole message as part of $\FF_q^l\simeq\FF_{q^l}$.
We also assume a set of nodes $R_1,\ldots,R_V$ which can verify the
authentication. A priori, this set can include the destinations, but can
also be larger. Typically we will assume that $V >>T$, in the context
of pollution attacks.

We now present our $(k,V,M)$ network coding authentication scheme and discuss 
its efficiency. Security will be analyzed in the next section.


\subsection{Set-up and authentication tag generation}

We propose the following authentication scheme:
\begin{enumerate}
\item
{\bf Key generation}: A trusted authority randomly generates $M+1$
polynomials $P_0(x),\ldots,P_M(x) \in \FF_{q^l}[x]$ and chooses $V$ distinct
values $x_1,\ldots,x_V \in \FF_{q^l}$. These polynomials are of degree $k-1$,
and we denote them by
\[
P_i(x)=a_{i0}+a_{i1}x+a_{i2}x^2+\ldots+a_{i,k-1}x^{k-1},
\]
$i=0,\ldots,M$.
\item {\bf Key distribution}: The trusted
authority gives as private key to the source $S$ the $M+1$
polynomials $(P_0(x),\ldots,P_M(x))$, and as private key for each
verifier $R_i$ the $M+1$ polynomials evaluated at $x=x_i$, namely
$(P_0(x_i),\ldots,P_M(x_i))$, $i=1,\ldots,V$. The values
$x_1,...,x_V$ are made public.
The keys can be given to the nodes at the same time as they are
given their local encoding vectors.

\item {\bf Authentication tag}: Let us assume that the source wants to
send $n$ data messages $s_1,\ldots,s_n\in\FF_q^l$. The source computes the 
following polynomial
\[
A_{s_i}(x)=P_0(x)+s_iP_1(x)+s_i^qP_2(x)\ldots+s_i^{q^{(M-1)}}P_M(x)\in\FF_{q^l}[x]
\]
which forms the authentication tag of each $s_i$, $i=1,\ldots,n$. 
The packets $\xv_i$ to be actually sent by the source are of the form
\[
\xv_i = [1,s_i,A_{s_i}(x)]\in\FF_q^{1+l+kl},~i=1,\ldots,n.
\]
The tag is attached after the message, and 1 bit is added at the beginning, 
which will be used to keep track of the network coding coefficients.
\end{enumerate}

The number $M+1$ of polynomials $P_i(x)$ is related to the number of usages
of the key, while the degree $k-1$ corresponds to the size of attackers
coalition.

Note that while making public the values $x_1,\ldots,x_V$ still may help
an attacker, we prefer to make them public and prove that actually this
does not help the attacker, in order to minimize the amount of secret
information given to the nodes.


\subsection{Verification and correctness of the authentication tag}

In order to discuss the authentication check, let us recall from
(\ref{eq:v_i}) what is the received tagged vector at a node $R_i$ with 
$i_h$ incoming edges when the source is sending 
$\xv_j=[1,s_j,A_{s_j}(x)]\in\FF_q^{1+l+kl}$, $j=1,\ldots,n$:
\begin{eqnarray*} \left(
\begin{array}{c}
\yv(e_{i_1})\\
\vdots\\
\yv(e_{i_h})
\end{array}
\right)
&=&
\left(
\begin{array}{ccc}
g_1(e_{i_1}) & \ldots & g_n(e_{i_1}) \\
\vdots & & \vdots \\
g_1(e_{i_h}) & \ldots & g_n(e_{i_h}) \\
\end{array}
\right)
\left(
\begin{array}{ccc}
1 & s_1 & A_{s_1}(x) \\
\vdots\\
1 & s_n & A_{s_n}(x)
\end{array}
\right)
\\
&=&
\left(\!\!
\begin{array}{ccc}
\sum_{j=1}^ng_j(e_{i_1})& \sum_{j=1}^ng_j(e_{i_1})s_j & \sum_{j=1}^ng_j(e_{i_1})A_{s_j}(x)\\
 & \vdots & \vdots \\
\sum_{j=1}^ng_j(e_{i_h}) & \sum_{j=1}^ng_j(e_{i_h})s_j & \sum_{j=1}^ng_j(e_{i_h})A_{s_j}(x)\\
\end{array}
\!\!\right).
\end{eqnarray*}
Recall that a verifying node $R_i$ further has a private key given by
\[
P_0(x_i),\ldots,P_M(x_i).
\]
For each incoming edge $e_k$, $k=i_1,\ldots,i_h$, the node $R_i$ can thus 
compute the product of the received data on the edge by the private keys, 
as follows:
\[
P_0(x_i)\sum_{j=1}^ng_j(e_k),~P_1(x_i)\sum_{j=1}^ng_j(e_k)s_j
\]
and similarly for the key $P_2(x_i)$
\[
P_2(x_i)\left(\sum_{j=1}^ng_j(e_k)s_j\right)^q=P_2(x_i)\sum_{j=1}^n(g_j(e_k)s_j)^q
=P_2(x_i)\sum_{j=1}^ng_j(e_k)s_j^q
\]
and the other keys $P_j(e_k)$, $j=3,\ldots,M$. For example:
\[ P_M(x_i)\left(\sum_{j=1}^ng_j(e_k)s_j\right)^{q^{(M-1)}}
=P_M(x_i)\sum_{j=1}^n(g_j(e_k)s_j)^{q^{(M-1)}}
=P_M(x_i)\sum_{j=1}^ng_j(e_k)s_j^{q^{M-1}}.
\]
On the other hand, it can evaluate the polynomial
\[
\sum_{j=1}^ng_j(e_k)A_{s_j}(x)
\]
in $x_i$, which is public. This yields
\begin{eqnarray*}
\sum_{j=1}^ng_j(e_k)A_{s_j}(x_i)
&=&\sum_{j=1}^ng_j(e_k)(P_0(x_i)+s_jP_1(x_i)+s_j^qP_2(x_i)+\ldots+s_j^{q^{(M-1)}}P_M(x_i))\\
&=&\sum_{j=1}^ng_j(e_k)P_0(x_i)+\sum_{j=1}^ng_j(e_k)s_jP_1(x_i)+\ldots+\sum_{j=1}^ng_j(e_k) s_j^{q^{(M-1)}}P_M(x_i).
\end{eqnarray*}
The node $R_i$ accepts a packet on its incoming edge $e_k$ if the two
computations coincide, which we have just shown they do if there is no 
alteration of the protocol. Note that the verifying node does not need to 
decode the message (which it may not be able to do) in order to perform 
the check.

\begin{example}\label{ex:net2}\rm
Consider the network of Example \ref{ex:net1} with a $(2,2,3)$ authentication 
scheme, where we have $V=2$ nodes which verify the authentication tags, say 
the relay node $R_1$ and the destination $D_1$, and the key can be used $3$ 
times, to protect against a coalition of at most $2$ attackers (either only 
$R_2$, or $R_2$ and $R_1$ if the latter gets corrupted though it has a private 
key). The source $S$ wants to send two messages 
$s_1,s_2 \in \FF_{2^3}\simeq\FF_2^3$, that is 
$s_1=(s_{1,1},s_{1,2},s_{1,3})\in\FF_2^3$ and 
$s_2=(s_{2,1},s_{2,2},s_{2,3})\in\FF_2^3$ with $s_{i,j}\in \FF_2=\{0,1\}$.
During the key generation and distribution, we have that:
\begin{itemize}
\item
The source is given the $M+1=3$ polynomials $P_0(x)=a_{00}+a_{01}x$,
$P_1(x)=a_{10}+a_{11}x$, and $P_2(x)=a_{20}+a_{21}x$, of degree $k-1=1$, 
with coefficients $a_{ij}$ in $\FF_{2^3}$.
\item
The values $x_1,x_2\in\FF_{2^3}$ are made public.
\item
The relay node $R_1$ receives the secret values $P_0(x_1)$, $P_1(x_1)$, 
$P_2(x_1)$ as its private key.
\item
The destination node $D_1$ receives the secret values $P_0(x_2)$, $P_1(x_2)$, 
$P_2(x_2)$ as its private key.
\end{itemize}

The source computes two authentication tags:
\begin{eqnarray*}
A_{s_1}(x)\!\!\!\!&=&\!\!\!\!
P_0(x)+s_1P_1(x)+s_1^2P_2(x)\\
&=&
(a_{00}+a_{10}s_1+a_{20}s_1^2)+x(a_{01}+a_{11}s_1+a_{21}s_1^2)\\
&=:&b_{10}+xb_{11},~b_{10},b_{11}\in\FF_8\\
A_{s_2}(x)\!\!\!\!&=&\!\!\!\!
P_0(x)+s_2P_1(x)+s_2^2P_2(x)\\
&=&
(a_{00}+a_{10}s_2+a_{20}s_2^2)+x(a_{01}+a_{11}s_2+a_{21}s_2^2)\\
&=:&b_{20}+xb_{21},~b_{20},b_{21}\in\FF_8.
\end{eqnarray*}
The two packets to be sent are
\begin{eqnarray*}
\xv_1 &=& [1,s_1,A_{s_1}(x)]=[1,s_{1,1},s_{1,2},s_{1,3},b_{10},b_{11}]
\in (\FF_2)^{10}\\
\xv_2 &=& [1,s_2,A_{s_2}(x)]=[1,s_{2,1},s_{2,2},s_{2,3},b_{20},b_{21}]
\in (\FF_2)^{10}.
\end{eqnarray*}
The first node $R_1$ has two input edges $e_1,e_2$, and its received vector is
given by
\begin{eqnarray*}
\left(
\begin{array}{c}
\yv(e_1)\\
\yv(e_2)
\end{array}
\right)&=&
\left(
\begin{array}{cc}
g_1(e_1)&g_2(e_1)\\
g_1(e_2)&g_2(e_2)
\end{array}
\right)
\left(
\begin{array}{ccc}
1& s_1 &A_{s_1}(x)\\
1& s_2 &A_{s_2}(x) \\
\end{array}
\right)\\
&=& \left(
\begin{array}{ccc}
g_1(e_1)+g_2(e_1)&g_1(e_1)s_1+g_2(e_1)s_2&g_1(e_1)A_{s_1}(x)+g_2(e_1)A_{s_2}(x)\\
g_1(e_2)+g_2(e_2)&g_1(e_2)s_1+g_2(e_2)s_2&g_1(e_2)A_{s_2}(x)+g_2(e_2)A_{s_2}(x) \\
\end{array}
\right).
\end{eqnarray*} 
The data which is public is $x_1,x_2$. Using $x_1$ and its private key 
$(P_0(x_1)$, $P_1(x_1)$, $P_2(x_1))$, $R_1$ can compute from $\yv(e_1)$ the 
following three terms:
\[
P_0(x_1)(g_1(e_1)+g_2(e_1)),~P_1(x_1)(g_1(e_1)s_1+g_2(e_1)s_2),
\]
and
\[
P_2(x_1)(g_1(e_1)s_1+g_2(e_1)s_2)^2=P_2(x_1)(g_1(e_1)^2s_1^2+g_2(e_1)^2s_2^2)
=P_2(x_1)(g_1(e_1)s_1^2+g_2(e_1)s_2^2)
\]
whose sum gives
\begin{equation}\label{eq:exsum}
P_0(x_1)(g_1(e_1)+g_2(e_1))+P_1(x_1)(g_1(e_1)s_1+g_2(e_1)s_2)+P_2(x_1)(g_1(e_1)s_1^2+g_2(e_1)s_2^2).
\end{equation}
Since $R_1$ has also received $g_1(e_1)A_{s_1}(x)+g_2(e_1)A_{s_2}(x)$, it can 
evaluate the polynomial in $x_1$ and check whether 
$g_1(e_1)A_{s_1}(x_1)+g_2(e_1)A_{s_2}(x_1)$ is equal to the sum (\ref{eq:exsum}).
If yes, the node $R_1$ accepts the authentication tag and re-encode the 
packet, otherwise, the packet is discarded.
A similar check is performed on $e_2$, and by the destination on its incoming 
edges using its own private key.
\end{example}


\subsection{Parameters and efficiency}

We discuss the efficiency of the proposed scheme, based on the communication, 
computation, and storage costs. The different parameters involved are 
summarized in Table \ref{tab:nbsymb}.

\begin{table}
\centering
\begin{tabular}{|c|c|c|c|}
\hline
parameters         &  notation          & symb/item & total in $\FF_q$\\
\hline
source private keys   & $P_i(x)$, $i=0,\ldots,M$    & $k$  & $k(M+1)l$ \\
public values & $x_i$, $i=1,\ldots,V$      & $1$  & $Vl$ \\
verifiers' private keys &$P_i(x_i)$, $i=0,\ldots,M$& $M+1$&$V(M+1)l$\\
tags               & $A_{s_i}(x)$, $i=1,\ldots,n$& $kl$ & $nkl$ \\
\hline
\end{tabular}
\normalsize 
\caption{Sizes for the keys and tags of the proposed $(k,V,M)$ scheme.}
\label{tab:nbsymb}
\end{table}

There are two classes of parameters, those fixed by the network, namely,
the number $T$ of destination nodes, the network code alphabet $\FF_q$, 
the length $l$ of the data packets, and $n$ the number of messages to be 
sent by the source. We then have the security parameters $k,V$ and $M$, 
which first depend on the network parameters: 
\begin{itemize}
\item {\bf Constraints on $V$}:
We will typically take $V >> T$, which means that more nodes than just
the destinations will check the authentication tags. We could imagine $V<T$ if
we do not even want all the destinations to check the authentication of their 
packets. However our goal is to have enough nodes in the network (though not
necessarily all of them) verifying the integrity of the packets to avoid
the propagation of polluted packets.
We further have $q^l \geq V$, since private verification keys are obtained by
evaluating the polynomials in $x_i$, $i=1,\ldots,V$. If $V \geq q^l$, then we
are forced to use some values of $\FF_q$ more than once, and the private
keys are not unique anymore. Thus $q^l \geq V >> T$.
\item {\bf Constraints on $M$}:
We assume that $M$ is at least greater than $n$, to be able to protect  with the
same key all the messages to be sent within one encoding round.
\end{itemize}

The scheme communication cost mainly relies on the size of the authentication 
tag $|A_{s_i}|$, $i=1,\ldots,n$, which is $O(kl)$, since the length of the tag 
is $kl$, and we also have to consider the augmentation of the data
vectors by one symbol element performed at the source.

The computational costs involve computing and appending the tag at the source, 
and verifying the tag at some intermediate nodes and at the destinations.
\begin{itemize}
\item{\bf Cost at the source}: For creating a tag based on a message 
$s_i\in\FF_{q^l}$, recall that the source computes the following polynomial: 
\[
A_{s_i}(x)=P_0(x)+s_iP_1(x)+s_i^qP_2(x)\ldots+s_i^{q^{(M-1)}}P_M(x)\in\FF_{q^l}[x],
\]
which involves thus $M-1$ exponentiations in $\FF_{q^l}$ to compute 
$s_i^{q^j}$, $j=1,\ldots,M-1$, and then $kM$ multiplications in $\FF_{q^l}$ 
to get $P_j(x)s_i^{q^{(j-1)}}$, $j=1,\ldots,M$. This is repeated for each of 
the $n$ messages $s_i$, $i=1,\ldots,n$.
\item{\bf Cost at the verifying nodes}: A verifying node $R_i$ needs to do two things to check the tag.
First, it computes
\[
P_0(x_i)\sum_{j=1}^ng_j(e_k),~P_1(x_i)\sum_{j=1}^ng_j(e_k)s_j,\ldots,
~P_M(x_i)\left(\sum_{j=1}^ng_j(e_k)s_j\right)^{q^{M-1}}
\]
which takes $M-1$ exponentiations in $\FF_{q^l}$ and $M+1$ multiplications 
in $\FF_{q^l}$, before evaluating the polynomial arrived on its incoming 
edge $e_k$
\[
\sum_{j=1}^ng_j(e_k)A_{s_j}(x) \in \FF_{q^l}[x] 
\]
in the public key $x_i\in\FF_{q^l}$. Since the polynomial is of degree $k-1$, 
its evaluation requires $k-2$ exponentiations in $\FF_q$ for $x_i^j$, 
$j=2,\ldots,k-1$, and $k-1$ multiplications in $\FF_{q^l}$ to multiply each 
$x_i^j$, $j=1,\ldots,k-1$, with the coefficients of the polynomial. This is 
done for each of the $i_h$ incoming edges.
\end{itemize}

Finally, the storage cost consists of the size of the keys, that is $M+1$ 
keys of size $k$ for the source, and the $M+1$ polynomials evaluated in one 
value of $\FF_{q^l}$, yielding $M+1$ values in $\FF_{q^l}$ for each of the verifying 
nodes.

All the costs of the proposed scheme are summarized in Table 
\ref{tab:Comparison}.

\begin{table}
\centering
\begin{tabular}{|c|c|c|c|c|c|c|c|}
\hline
Tag or signature size&$kl$\\
\hline
Communication cost &$kl+1$\\
\hline
Tag or signature &$n(M-1)l$ exp \\
computational cost &$nkMl$ mult \\
\hline
Verification & $((M-1)+k-2)lh$ exp \\
computational cost &$((M+1)+k-1)lh$ mult \\
\hline
Storage at the source& $(M+1)lk $\\
Storage at the verifiers& $(M+1)l$ \\
\hline
\end{tabular}
\normalsize
\caption{Efficiency of the proposed scheme. The parameter $h$ denotes the 
number of incoming edges of a verifying node. Operations (multiplications 
and exponentiations) as well as numbers of symbols are in $\FF_q$.}
\label{tab:Comparison}
\end{table}

%
%

\section{Security Analysis of the Authentication Scheme}
\label{sec:analysis}

Threats are coming from either outside or inside opponents, who can attempt 
either impersonation or substitution attacks. Outside opponents are assumed 
to be able to see the data on the incoming edges of some of the intermediate 
nodes. Inside opponents of course see the messages transiting through them, 
but the difference is that some of them may actually be verifying nodes, and 
thus they can use their own private keys to forge a substitution attack. 
The analysis focuses on the worst case scenario, namely a coalition of inside 
malicious nodes in possession of private keys is trying to make a substitution 
attack, that is, to send a fake packet after observing tagged messages in such 
a way that a node which checks for authentication will actually accept the 
faked authentication tag. 


\subsection{Preliminaries}

In the following, we may write as matrix indices the dimension of the 
matrices for clarity. 

Suppose that a malicious node has $i_h$ incoming edges, with
received vector
\begin{eqnarray*}
\left(
\begin{array}{c}
\yv(e_{i_1})\\
\vdots\\
\yv(e_{i_h})
\end{array}
\right)
&=&
\left(
\begin{array}{ccc}
g_1(e_{i_1}) & \ldots & g_n(e_{i_1}) \\
\vdots & & \vdots \\
g_1(e_{i_h}) & \ldots & g_n(e_{i_h}) \\
\end{array}
\right)
\left(
\begin{array}{ccc}
1 & s_1 & A_{s_1}(x) \\
\vdots&\vdots& \vdots\\
1 & s_n & A_{s_n}(x)
\end{array}
\right)
\\
&=&
\left(
\begin{array}{ccc}
\sum_{j=1}^ng_j(e_{i_1}) & \sum_{j=1}^ng_j(e_{i_1})s_j & \sum_{j=1}^ng_j(e_{i_1})A_{s_j}(x)\\
 \vdots& \vdots & \vdots \\
\sum_{j=1}^ng_j(e_{i_h}) & \sum_{j=1}^ng_j(e_{i_h})s_j & \sum_{j=1}^ng_j(e_{i_h})A_{s_j}(x)\\
\end{array}
\right),
\end{eqnarray*}
from which it tries to learn about the source private keys.
If we write
\begin{eqnarray*}
A_{s_j}(x)
&=& P_0(x)+s_jP_1(x)+\ldots+s_j^{q^{M-1}}P_M(x) \\
&=&b_{j0}+b_{j1}x+\ldots+b_{j,k-1}x^{k-1}\in\FF_{q^l}[x],
\end{eqnarray*}
we have that for all incoming edges $e_m$
\begin{eqnarray*}
\sum_{j=1}^ng_j(e_m)A_{s_j}(x)
&=&
\sum_{j=1}^ng_j(e_m)(b_{j0}+b_{j1}x+\ldots+b_{j,k-1}x^{k-1})\\
&=&
c_{m0}+c_{m1}x+\ldots+c_{m,k-1}x^{k-1},
\end{eqnarray*}
where
\[
c_{mi}=\sum_{j=1}^ng_j(e_m)b_{ji}\in\FF_{q^l}.
\]
Thus, the malicious node actually knows $c_{mi}$, $i=1,\ldots,k-1$, for every 
incoming edge $e_m$, $m=i_1,\ldots,i_h$, and upon reception of its incoming 
vector, it can obtain the following system of linear equations:
\begin{equation}\label{eq:systeq}
A_{k\times (M+1)}G_{(M+1)\times h}=C_{k\times h}.
\end{equation}
Both the matrix $G$ containing the network
coding coefficients and the matrix $C$ respectively given by
\[
C=
\left(
\begin{array}{ccc}
c_{10}& \ldots & c_{h,0}\\
\vdots& & \vdots \\
c_{1,k-1}&\ldots & c_{h,k-1}
\end{array}
\right),~
G=
\left(\!\!\!
\begin{array}{ccc}
\sum_{j=1}^ng_j(e_{i_1}) & \hdots & \sum_{j=1}^ng_j(e_{i_h})\\
\sum_{j=1}^ng_j(e_{i_1})s_j &\hdots & \sum_{j=1}^ng_j(e_{i_h})s_j\\
\vdots&& \vdots\\
\sum_{j=1}^ng_j(e_{i_1})s_j^{q^{(M-1)}} &\hdots&\sum_{j=1}^ng_j(e_{i_h})s_j^{q^{(M-1)}}
\end{array}
\!\!\!\right)
\]
are known to the malicious node, while the $k\times (M+1)$ matrix $A$ given by 
\[
A=
\left(
\begin{array}{cccc}
a_{0,0} & a_{1,0} & \hdots  & a_{M,0}\\
a_{0,1} & a_{1,1} &  & a_{M,1} \\
  \vdots   & \vdots        &  & \vdots   \\
a_{0,k-1} & a_{1,k-1} & \hdots & a_{M,k-1}
\end{array}
\right)
\]
is to be found. 
$A$ has on its $i$th column the coefficients $a_{i0},\ldots,a_{i,k-1}$ of the 
$i$th secret polynomial $P_i$, and thus 
contains all the coefficients of the source's private keys.

Let us now assume that $K$ nodes collaborate to make a
substitution attack. Each of them first obtains vectors of data from the 
network, and can thus collect a system of linear equations of the form
\[
AG_i=C_i,~i=1,\ldots,K,
\]
as explained in (\ref{eq:systeq}). The number of columns of $G_i$ depends on 
the number of incoming edges $h_i$ at the $i$th corrupted node. 
All together, this gives a new system of linear equations of the form
\[
A_{k\times (M+1)}\Gc_{(M+1)\times (h_1+\ldots+h_K)}=\Cc_{k\times (h_1+\ldots+h_K)}
\]
with
\[
\Gc=[G_1~G_2~\ldots G_K],~\Cc=[C_1,\ldots,C_K]
\]
where all matrices $\Gc$, $\Cc$ and $A$ have coefficients in $\FF_{q^l}$.

We now take into account that some of the nodes who are given the
private keys to check the authentication could be corrupted. Since we assume
a group of $K$ malicious nodes, let us furthermore assume the worst case,
namely that all of them actually possess a private key
$(P_0(x_i),\ldots,P_M(x_i))$, where $i$ belongs to a subset of cardinality
$K$ of $\{1,\ldots,V\}$. Without loss of generality we can assume that
$i$ goes from $1$ to $K$.

Since the values $x_1,\ldots,x_V$ are made public, the group of adversaries
can actually build another system of linear equations which exploits their
knowledge of the private keys, namely
\[
X_{K\times k}A_{k\times (M+1)}=P_{K\times (M+1)}
\]
where
\[
X=
\left(
\begin{array}{cccc}
1 & x_1 & \ldots & x_1^{k-1}\\
1 & x_2 & \ldots & x_2^{k-1}\\
\vdots & \vdots & & \vdots \\
1 & x_K & \ldots & x_K^{k-1}
\end{array}
\right)
\]
contains the public key values, as before
\[
A=
\left(
\begin{array}{cccc}
a_{0,0} & a_{1,0} & \ldots & a_{M,0}\\
a_{0,1} & a_{1,1} & \ldots & a_{M,1} \\
 \vdots      &  \vdots       &  & \vdots   \\
a_{0,k-1} & a_{1,k-1} & \ldots & a_{M,k-1}
\end{array}
\right)
\]
contains the coefficients 
of the private key to be found by the group of attackers,
and
\[
P=
\left(
\begin{array}{cccc}
P_0(x_1) & P_1(x_1) & \ldots & P_M(x_1)\\
P_0(x_2) & P_1(x_2) & \ldots & P_M(x_2)\\
   \vdots    & \vdots        &  &  \vdots  \\
P_0(x_K) & P_1(x_K) & \ldots & P_M(x_K)\\
\end{array}
\right),
\]
contains the private keys of the corrupted nodes.

Since the polynomials $P_0,\ldots,P_M$ have degree $k-1$, it is clear
that $K$ can be at most $k-1$, otherwise from the knowledge of only
the private and public keys, the group of attackers can recover the source's
private key, i.e., they can solve the system of equations and
recover $A$.

By putting together the information given by the private keys and the one 
gathered from all the received vectors, the group of adversaries has now the 
knowledge of the following linear systems of equations for trying to find the 
source private key:
\[
A_{k\times (M+1)}\Gc_{(M+1)\times H}=\Cc_{k\times H},
~X_{K\times k}A_{k\times (M+1)}=P_{K\times (M+1)},
\]
where $H=h_1+\ldots+h_K$ is the aggregated number of incoming edges for 
all corrupted nodes and $K\leq k-1$.


\subsection{Main analysis}

Let us start this part by proving some technical lemmas.

\begin{lemma}\label{lem:poly}
Consider the finite field $\FF_{q'}$ and the polynomial $F(x,y)$ in 
$\FF_{q'}[x,y]$ given by
\[
F(x,y)=(x-\alpha_1)\ldots(x-\alpha_Q)(y-\beta_1)\ldots(y-\beta_R)
\]
of degree $Q$ in $x$ and $R$ in $y$. Then there exists a $(Q+1)\times(R+1)$ 
matrix $A$ such that
\[
\left(
\begin{array}{cccc}
1 & \alpha_1 & \ldots & \alpha_1^Q\\
1 & \alpha_2 & \ldots & \alpha_2^Q\\
\vdots & & &  \vdots\\
1 & \alpha_q & \ldots & \alpha_q^Q\\
\end{array}
\right)A={\bf 0}_{q\times (R+1)}\mbox{ and }
A
\left(
\begin{array}{cccc}
1 & 1 & \ldots & 1 \\
\beta_1 & \beta_2 & \ldots & \beta_r\\
\vdots & & &  \vdots\\
\beta_1^R & \beta_2^R & \ldots & \beta_r^R
\end{array}
\right)={\bf 0}_{(Q+1)\times r},
\]
for $1\leq q\leq Q$ and $1\leq r \leq R$.
\end{lemma}
\begin{proof}
Let us develop the products in $x$ and $y$ of
$F(x,y)=(x-\alpha_1)\ldots(x-\alpha_Q)(y-\beta_1)\ldots(y-\beta_R)$ 
respectively to get
\[
a(x)=(x-\alpha_1)\ldots(x-\alpha_Q)=a_0+a_1x+\ldots+a_Qx^Q
\]
and
\[
b(y)=(y-\beta_1)\ldots(y-\beta_R)=b_0+b_1y+\ldots+b_Ry^R.
\]
Now we can write
\[
F(x,y)=a(x)b(y)=(1,x,\ldots,x^Q)
\underbrace{
\left(
\begin{array}{c}
a_0\\
a_1 \\
\vdots\\
a^Q
\end{array}
\right)
(b_0,b_1,\ldots,b_R)}_A
\left(
\begin{array}{c}
1\\
y \\
\vdots\\
y^R
\end{array}
\right)
\]
for the matrix $A$ with coefficients in $\FF_{q'}$.
Since $F(\alpha_q,y)=0$ for $1\leq q \leq Q$, we have that
\[
F(\alpha_q,y)=(1,\alpha_q,\ldots,\alpha_q^Q)A
\left(
\begin{array}{c}
1\\
y \\
\vdots\\
y^R
\end{array}
\right)=0
\] 
for all $y$ which proves the first equality. 
The claim follows similarly by using that $F(x,\beta_r)=0$ for
$1\leq r \leq R$.
\end{proof}

\begin{example}\rm
Take
\begin{eqnarray*}
F(x,y)&=&(x-\alpha_1)(y-\beta_1)(y-\beta_2) \\
      &=&(x-\alpha_1)(y^2+y(-\beta_1-\beta_2)+\beta_1\beta_2).
\end{eqnarray*}
We have that
\begin{eqnarray*}
F(x,y)&=&
(1,x)
\left(
\begin{array}{c}
-\alpha_1 \\
1
\end{array}\right)
(\beta_1\beta_2,-\beta_1-\beta_2,1)
\left(
\begin{array}{c}
1 \\
y \\
y^2
\end{array}\right)\\
&=&
(1,x)
\underbrace{
\left(
\begin{array}{ccc}
-\alpha_1\beta_1\beta_2 & \alpha_1(\beta_1+\beta_2)&-\alpha_1 \\
\beta_1\beta_2 &-(\beta_1+\beta_2) & 1
\end{array}
\right)}_{A}
\left(
\begin{array}{c}
1 \\
y \\
y^2
\end{array}
\right).
\end{eqnarray*}
Thus
\[
F(\alpha_1,y)=(1,\alpha_1)A
\left(
\begin{array}{c}
1 \\
y \\
y^2
\end{array}
\right)=0
\]
and 
\[
(1,\alpha_1)A=(0,0,0).
\]
\end{example}

\begin{lemma}\label{lem:newpoly}
Consider the finite field $\FF_{q'}$.
\begin{enumerate}
\item
Let
\[
b(y)=b_0+b_1y+b_2y^2+\ldots+b_qy^q+\ldots+b_{q^{M-1}}y^{q^{M-1}}
\]
be a polynomial in $\FF_{q'}[y]$. 
If all the coefficients $b_i$ are zero, but for the $M+1$ coefficients $b_0$ 
and $b_{q^j}$, $j=0,\ldots,M-1$ which can take any values in $\FF_{q'}$, then 
for all choices of $\gamma_1,\ldots,\gamma_H$ in $\FF_{q'}$, there exists a 
polynomial $c(y)\in \FF_{q'}[y]$ of degree $q^{M-1}-H$ such that 
\[
b(y)=(y-\gamma_1)\cdots(y-\gamma_H)c(y)
\]  
provided that $H\leq M$.
\item
Consider the polynomial $F(x,y)$ in $\FF_{q'}[x,y]$ given by
\[
F(x,y)=(x-\alpha_1)\cdots(x-\alpha_Q)b(y)
\]
of degree $Q$ in $x$ and where $b(y)=b_0+b_1y+b_2y^q+\ldots+b_My^{q^{M-1}}$ 
is as above, in particular it is of degree $q^{M-1}$ and has 
$\gamma_1,\ldots,\gamma_H \in \FF_{q'}$ as roots. Then there exists a 
$(Q+1)\times(M+1)$ matrix $A$ such that
\[
\left(
\begin{array}{cccc}
1 & \alpha_1 & \ldots & \alpha_1^Q\\
1 & \alpha_2 & \ldots & \alpha_2^Q\\
\vdots\\
1 & \alpha_q & \ldots & \alpha_q^Q\\
\end{array}
\right)A={\bf 0}_{q\times (M+1)}\mbox{ and }
A
\left(
\begin{array}{cccc}
1 & 1 & \ldots & 1 \\
\gamma_1 & \gamma_2 & \ldots & \gamma_H\\
\gamma_1^q & \gamma_2^q & \ldots & \gamma_H^q\\
\vdots\\
\gamma_1^{q^{M-1}} & \gamma_2^{q^{M-1}} & \ldots & \gamma_H^{q^{M-1}}
\end{array}
\right)={\bf 0}_{(Q+1)\times H},
\]
for $1\leq q\leq Q$ and $1\leq H \leq M$.
\end{enumerate}
\end{lemma}
\begin{proof}
\begin{enumerate}
\item
Consider the polynomial 
\[
b(y)=b_0+b_1y+b_2y^2+\ldots+b_qy^q+\ldots+b_{q^{M-1}}y^{q^{M-1}}
\]
where all the coefficients $b_i$ are zero, but for the $M+1$ coefficients 
$b_0$ and $b_{q^j}$, $j=0,\ldots,M-1$ which can take any values in $\FF_{q'}$. 
For all choices of $\gamma_1,\ldots,\gamma_H$ 
in $\FF_{q'}$, we can form the polynomial $d(y)$ by defining
\[
d(y)=(y-\gamma_1)(y-\gamma_2)\cdots(y-\gamma_H).
\]
What we claim is that, provided that $H\leq M$, there exists a polynomial 
$c(y)\in \FF_{q'}[y]$ such that 
\[
b(y)=(y-\gamma_1)\cdots(y-\gamma_H)c(y)=d(y)c(y),
\]  
or in other words, we can choose $c(y)$ such that 
\[
(y-\gamma_1)\cdots(y-\gamma_H)c(y)
\]
is a polynomial whose coefficients are all zero but for $M+1$ of them, 
which are the constant term and the $q^j$th term, for $j=0,\ldots,M-1$.

To prove this, let us write $d(y)$ as
\[
d(y)=d_0+d_1y+d_2y^2+\ldots+d_Hy^H.
\]
The equation $d(y)c(y)=b(y)$ can be rewritten, by identifying the coefficients 
of $y$, as
\[
\underbrace{
\left(
\begin{array}{cccc}
d_0 & 0   & & \\
d_1 & d_0 & & \\
\vdots & d_1 & \ddots& \\
d_H & \vdots &  &   d_0\\
0 & d_H & &  \\
   &  & \ddots &  \\
   &  &         & d_H\\
\end{array}
\right)}_{D \mbox{ {\footnotesize of size} }(q^{M-1}+1)\times(q^{M-1}-H+1)}
\left(
\begin{array}{c}
c_0 \\
c_1 \\
\vdots \\
c_{q^{M-1}-H}
\end{array}
\right)
=
\left(
\begin{array}{c}
b_0 \\
b_1 \\
\vdots \\
b_{q^{M-1}}
\end{array}
\right).
\]
Among the $q^{M-1}+1$ coefficients $b_i$, we do not have any constraint 
on the constant term and the $q^j$th term $j=0,\ldots,M-1$, which can take 
any value. We only have as constraints that the other coefficients are zero. 
We thus care about $q^{M-1}+1-(M+1)=q^{M-1}-M$ of them, which means we can 
remove $M+1$ rows from both sides of the above system of equations. The matrix 
$D$ containing the coefficients $d_i$ is now a $(q^{M-1}-M)\times (q^{M-1}-H+1)$ 
matrix. Any 
wanted polynomial $c(y)$ corresponds to a vector $(c_0,\ldots,c_{q^{M-1}-H})$ 
which belongs to the kernel of $D$. For this vector to exist and be non-zero, 
we need the kernel of $D$ to be of dimension at least 1, for which the rank 
${\rm rk}(D)$ of $D$ must be smaller or equal to $q^{M-1}-H$. Now we have that 
\[
{\rm rk}(D)\leq \min (q^{M-1}-M,q^{M-1}-H+1).
\] 
Thus if $H \leq M$ as assumed, we get that
\[
{\rm rk}(D)\leq \min (q^{M-1}-M,q^{M-1}-H+1)=q^{M-1}-M \leq q^{M-1}-H
\] 
and we are done.
\item
As in the proof of Lemma \ref{lem:poly}, we first develop the product in $x$ 
from $F(x,y)=(x-\alpha_1)\cdots(x-\alpha_Q)b(y)$ to get
\[
a(x)=(x-\alpha_1)\ldots(x-\alpha_Q)=a_0+a_1x+\ldots+a_Qx^Q.
\]
Since $b(y)$ is given by
\[
b(y)=(y-\gamma_1)\ldots(y-\gamma_H)c(y)=b_0+b_1y+b_2y^q\ldots+b_My^{q^{M-1}},
\]
we can write
\[
F(x,y)=(1,x,\ldots,x^Q)
\underbrace{
\left(
\begin{array}{c}
a_0\\
a_1 \\
a_2\\
\vdots\\
a_Q
\end{array}
\right)
(b_0,b_1,b_2,\ldots,b_M)}_A
\left(
\begin{array}{c}
1\\
y \\
y^q\\
\vdots\\
y^{q^{M-1}}
\end{array}
\right)
\]
for the matrix $A$ with coefficients in $\FF_{q'}$.
Since $F(\alpha_q,y)=0$ for $1\leq q \leq Q$, we have that
\[
F(\alpha_q,y)=(1,\alpha_q,\ldots,\alpha_q^Q)A
\left(
\begin{array}{c}
1\\
y \\
y^q\\
\vdots\\
y^{q^{M-1}}
\end{array}
\right)=0
\] 
for all $y$ which proves the first equality.
The claim follows similarly by using that $F(x,\gamma_H)=0$ for
$1\leq H \leq M$, by the previous point of the lemma. In words, the number 
of rows and columns of the matrix $A$ are decided by the number of 
(non-zero) coefficients in the polynomial $a(x)$ and $b(y)$ respectively. 
On the other hand, the number of rows of the matrices with coefficients in 
$\gamma$ and in $\alpha$ depends on the number of roots of the respective 
polynomials.  
\end{enumerate}
\end{proof}

\begin{example}\rm
\begin{itemize}
\item
Take first $q=2$ and $M=3$. For any choice of $\gamma_1,\gamma_2,\gamma_3$, we 
can define
\[
d(y)=(y-\gamma_1)(y-\gamma_2)(y-\gamma_3).
\]
Now since
\[
b(y)=b_0+b_1y+b_2y^2+b_3y^4,
\]
this means that we are looking for a linear polynomial 
\[
c(y)=y-\gamma_4.
\]
It is easy to see here that we can choose $\gamma_4=\gamma_1+\gamma_2+\gamma_3$.
\item
Take $q=2$ and $M=4$. We have for any choice of 
$\gamma_1,\gamma_2,\gamma_3,\gamma_4$ that
\begin{eqnarray*}
d(y)&=&(y-\gamma_1)(y-\gamma_2)(y-\gamma_3)(y-\gamma_4)\\
    &=& d_0+d_1y+d_2y^2+d_3y^3+d_4y^4
\end{eqnarray*}
with
\begin{eqnarray*}
d_0 &=& \gamma_1\gamma_2 \gamma_3 \gamma_4 \\
d_1 &=& -\gamma_1\gamma_2\gamma_3 -\gamma_1\gamma_2\gamma_4 -\gamma_1 \gamma_3 \gamma_4  -\gamma_2 \gamma_3 \gamma_4 \\
d_2 &=& \gamma_1 \gamma_2+\gamma_1 \gamma_3+\gamma_2\gamma_3+\gamma_1\gamma_4 +\gamma_2 \gamma_4 +\gamma_3 \gamma_4\\
d_3 & =& -\gamma_1-\gamma_2 -\gamma_3- \gamma_4 \\
d_4 & = & 1.
\end{eqnarray*}
The polynomial $b(y)$ is given by 
\[
b(y)=b_0+b_1y+b_2y^2+b_3y^4+b_4y^8,
\]
and in order to find a polynomial $c(y)=c_0+c_1y+c_2y^2+c_3y^3+c_4y^4+c_5y^5$ 
such that $c(y)d(y)=b(y)$, we have to solve the following system of equations:
\[
\left(
\begin{array}{ccccc}
d_3 & d_2& d_1 & d_0 & 0\\
0   & d_4& d_3 & d_2 & d_1 \\
0   & 0  & d_4 & d_3 & d_2 \\
0   & 0  &  0   & d_4 & d_3
\end{array}
\right)
\left(
\begin{array}{c}
c_0 \\
c_1 \\
c_2 \\
c_3 \\
c_4 \\
c_5
\end{array}
\right)
=
\left(
\begin{array}{c}
0\\
0\\
0\\
0
\end{array}
\right).
\]
Clearly the dimension of the kernel is at least 1.
\end{itemize}
\end{example}

\begin{lemma}\label{lem:nbsol}
If $K\leq k-1$ and $H \leq M$, there exist $q^l$ matrices $A_{k\times (M+1)}$ 
with coefficients in $\FF_{q^l}$ such that
\[
A_{k\times (M+1)}\Gc_{(M+1)\times H}={\bf 0}_{k\times H},
~X_{K\times k}A={\bf 0}_{K\times (M+1)}.
\]
\end{lemma}
\begin{proof}
Let $A=A_{k\times (M+1)}$ be a solution to the above system of equations. 
Then the matrices $rA$ obtained by multiplication with a scalar $r\in\FF_{q^l}$ 
are also clearly solutions, and it is thus enough to show that there exists 
one suitable matrix $A$.

To prove that such a matrix exists, we use Lemma \ref{lem:newpoly}, for which 
we will exhibit a suitable bivariate polynomial $F(x,y)=a(x)b(y)$.

Let us start by looking at the second equation. For any choice of $K$ public 
keys $x_1,\ldots,x_K$ in $\FF_{q^l}$, take the polynomial 
$a(x)=(x-x_1)\ldots (x-x_K)=a_0+a_1x+\ldots+a_Kx^K$. It is of degree $K$ 
and has for roots $x_1,\ldots,x_K$. Thus 
\[
a(x)=
(1,x,\ldots,x^K)
\left(
\begin{array}{c}
a_0\\
a_1 \\
a_2\\
\vdots\\
a_K
\end{array}
\right)
=0 \mbox{ for }x=x_1,\ldots,x_K.
\]

We now consider the first equation $A\Gc={\bf 0}$. We start by rewriting it 
in a different form. Recall that the matrix $\Gc$ is of the form
\[
\left(
\begin{array}{ccc}
\sum_{j=1}^ng_j(e_{i_1}) &\hdots & \sum_{j=1}^ng_j(e_{i_H})\\
\sum_{j=1}^ng_j(e_{i_1})s_j &\hdots & \sum_{j=1}^ng_j(e_{i_H})s_j\\
&&\\
\sum_{j=1}^ng_j(e_{i_1})s_j^{q^{(M-1)}} &\hdots&\sum_{j=1}^ng_j(e_{i_H})s_j^{q^{(M-1)}}
\end{array}
\right).
\]
Note that for any invertible matrix $\Dc$, we have that
\[
A\Gc={\bf 0} \iff A\Gc\Dc={\bf 0},
\]
and there exists an invertible matrix $\Dc$ such that $\Gc\Dc$ is of the
Vandermonde like form
\begin{equation}\label{eq:vand}
\left(
\begin{array}{ccc}
1             &\hdots & 1 \\
\gamma_1      &\hdots &\gamma_H\\
\gamma_1^q    & &\gamma_H^q\\
              & & \\
\gamma_1^{q^{M-1}} &\hdots& \gamma_H^{q^{M-1}}
\end{array}
\right).
\end{equation}
Indeed, if all the coefficients of the first row of $\Gc$ are non zero,
we can take $\Dc$ to be
\[
\Dc=diag((\sum_{j=1}^ng_j(e_{i_1}))^{-1},\ldots,(\sum_{j=1}^ng_j(e_{i_H}))^{-1}),
\]
in which case we have
\[
\gamma_k= \frac{\sum_{j=1}^ng_j(e_{i_k})s_j}{\sum_{j=1}^ng_j(e_{i_k})}
\]
where the denominator is in $\FF_q$ since it only depends on the network coding 
coefficients.
If the $i$th coefficient (say the first for example) of the first row is zero, 
then we can first compute $\Gc\Sc$ with
\[
\Sc
=
\left(
\begin{array}{cccc}
1  & 0 &  & 0 \\
1  & 0 &  & 0 \\
0  & &{\bf I}_{H-2}&\\
\end{array}
\right)
\]
which yields
\[
\left(
\begin{array}{ccc}
\sum_{j=1}^ng_j(e_{i_2}) &\hdots & \sum_{j=1}^ng_j(e_{i_H})\\
\sum_{j=1}^ng_j(e_{i_1})s_j+\sum_{j=1}^ng_j(e_{i_2})s_j &\hdots & \sum_{j=1}^ng_j(e_{i_H})s_j\\
&&\\
\sum_{j=1}^ng_j(e_{i_1})s_j^{q^{(M-1)}}+\sum_{j=1}^ng_j(e_{i_2})s_j^{q^{(M-1)}} &\hdots&\sum_{j=1}^ng_j(e_{i_H})s_j^{q^{(M-1)}}
\end{array}
\right).
\]
Now take
\[
\Dc'=diag((\sum_{j=1}^ng_j(e_{i_2}))^{-1},(\sum_{j=1}^ng_j(e_{i_2}))^{-1},\ldots,(\sum_{j=1}^ng_j(e_{i_H}))^{-1})
\]
and $\Dc=\Sc\Dc'$ to finally obtain
\[
\gamma_1= \frac{\sum_{j=1}^ng_j(e_{i_1})s_j+\sum_{j=1}^ng_j(e_{i_2})s_j}{\sum_{j=1}^ng_j(e_{i_2})},~
\gamma_k= \frac{\sum_{j=1}^ng_j(e_{i_k})s_j}{\sum_{j=1}^ng_j(e_{i_k})},~k\geq 2.
\]
Thus we can assume that we look at
\[
A\Gc={\bf 0}
\]
with $\Gc$ of the form (\ref{eq:vand}).

Now for all choices of $\gamma_1,\ldots,\gamma_H$, consider the polynomial 
$b(y)$ in $\FF_{q^l}[y]$ such that $b(\gamma_i)=0$, $i=1,\ldots,H$, but also 
such that $b(y)=b_0+b_1y+b_2y^q+\ldots+b_My^{q^{(M-1)}}$. Such polynomial 
exists by the first part of Lemma \ref{lem:newpoly}, as long as $H\leq M$. 
Thus
\[
b(y)=
(1,y,y^q,\ldots,y^{q^{M-1}})
\left(
\begin{array}{c}
b_0\\
b_1 \\
b_2\\
\vdots\\
b_M
\end{array}
\right)
=0 \mbox{ for }y=\gamma_1,\ldots,\gamma_H.
\]

We can finally write
\[
F(x,y)=(1,x,\ldots,x^K)
\underbrace{
\left(
\begin{array}{c}
a_0\\
a_1 \\
a_2\\
\vdots\\
a_K
\end{array}
\right)
(b_0,b_1,b_2,\ldots,b_M)}_A
\left(
\begin{array}{c}
1\\
y \\
y^q\\
\vdots\\
y^{q^{M-1}}
\end{array}
\right)
\]
for the $(K+1)\times (M+1)$ matrix $A$ with coefficients in $\FF_{q'}$.
By the second part of Lemma \ref{lem:newpoly}, $A$ satisfies
\[
\left(
\begin{array}{cccc}
1 & x_1 & \ldots & x_1^K\\
1 & x_2 & \ldots & x_2^K\\
\vdots\\
1 & x_K & \ldots & x_K^K\\
\end{array}
\right)A
={\bf 0}_{K\times (M+1)}
\]
and
\[
A
\left(
\begin{array}{cccc}
1 & 1 & \ldots & 1 \\
\gamma_1 & \gamma_2 & \ldots & \gamma_h\\
\gamma_1^q & \gamma_2^q & \ldots & \gamma_h^q\\
\vdots\\
\gamma_1^{q^{M-1}} & \gamma_2^{q^{M-1}} & \ldots & \gamma_h^{q^{M-1}}
\end{array}
\right)={\bf 0}_{(K+1)\times (M+1)}.
\]
Since the matrix $X$ build by the adversary has $k$ columns, we 
require $K+1=k$, that is $K=k-1$. This is indeed the assumption that 
we made on $K$ and $k$ in the hypothesis, and this can be interpreted 
by the fact that if $K \geq k$, then the adversaries can find the 
source's secret just from the matrix $X$. 
This concludes the proof.
\end{proof}

We are now ready to state the security of the proposed authentication scheme.

\begin{prop}
Consider a multicast network implementing linear network coding, among 
which nodes $V$ of them are verifying nodes owning a private key for 
authentication. 
The above scheme is a $(k,V,M)$ unconditionally secure network coding 
authentication code against a coalition of up to $k-1$ adversaries, possibly 
among the verifying nodes, in which every key can be used to authenticate up 
to $M$ messages, under the assumption that $H \leq M$, where $H$ is the sum 
of the incoming edges at each adversary.
\end{prop}
\begin{proof}
To make a substitution attack, the malicious $k-1$ verifying nodes want to 
generate a message such that it is accepted as authentic by any honest 
verifying node $R_i$ that they are trying to cheat. However, for that, they 
need to guess its secret key $[P_0(x_i),\ldots,P_M(x_i)]$, and choose a 
polynomial $\tilde{A}_s(x)$ such that
\[
\tilde{A}_s(x_i)=P_0(x_i)+s^qP_1(x_i)+\ldots+s^{q^{(M-1)}}P_M(x_i)
\]
for some message $s$. Gathering all they know after watching one transmission 
of tagged messages, the coalition of adversaries 
get the following system of equations:
\[
A_{k\times (M+1)}\Gc_{(M+1)\times H}=\Cc_{k\times H},
~X_{K\times k}A=P_{K\times (M+1)}.
\]
If there is no matrix $A_{k\times (M+1)}$ satisfying this system, the 
information gathered by the adversaries is not useful. Now if such a matrix 
$A_{k\times (M+1)}$ indeed exist, then there are actually $q^l$ of them 
satisfying these equations, given by  
\[
A_{k\times (M+1)}+A'_{k\times (M+1)},
\]
where $A'=A'_{k\times (M+1)}$ is a solution of the corresponding homogeneous 
system of equations, and Lemma \ref{lem:nbsol} tells us that there are $q^l$ 
such $A'$.
Thus there are $q^l$ different $(M+1)$-tuple of polynomials 
$(\tilde{P}_0(x),\ldots,\tilde{P}_M(x))$ likely to be the source's private 
key, from which that there are $q^l$ equally likely private keys 
for $R_i$. Thus the probability of the $k-1$ receivers to guess $A(x_i)$ 
correctly is $1/q^l$.
\end{proof}

\begin{example}\label{ex:net3}\rm
Let us go on with Example \ref{ex:net2}. The node $R_1$ has received
the vector
\[
\left(
\begin{array}{c}
\yv(e_1)\\
\yv(e_2)
\end{array}
\right)
\]
with
\[ \yv(e_1)=(g_1(e_1)+g_2(e_1),g_1(e_1)s_1+g_2(e_1)s_2,g_1(e_1)A_{s_1}(x)+g_2(e_1)A_{s_2}(x))
\]
and
\[
\yv(e_2)=(g_1(e_2)+g_2(e_2),g_1(e_2)s_1+g_2(e_2)s_2,g_1(e_2)A_{s_2}(x)+g_2(e_2)A_{s_2}(x)).
\] But this time, let us assume that the node $R_1$ is malicious, and instead of
checking the authentication tag, it actually wants to make a substitution attack.

Since we have that
\begin{eqnarray*}
A_{s_1}(x)\!\!\!\!&=&\!\!\!\!
P_0(x)+s_1P_1(x)+s_1^2P_2(x)\\
&=&
(a_{00}+a_{10}s_1+a_{20}s_1^2)+x(a_{01}+a_{11}s_1+a_{21}s_1^2)\\
&=:&b_{10}+xb_{11}\\
A_{s_2}(x)\!\!\!\!&=&\!\!\!\!
P_0(x)+s_2P_1(x)+s_2^2P_2(x)\\
&=&
(a_{00}+a_{10}s_2+a_{20}s_2^2)+x(a_{01}+a_{11}s_2+a_{21}s_2^2)\\
&=:&b_{20}+xb_{21},
\end{eqnarray*}
we can rewrite
\begin{eqnarray*}
g_1(e_1)A_{s_1}(x)+g_2(e_1)A_{s_2}(x)
&=&
g_1(e_1)(b_{10}+xb_{11})+g_2(e_1)(b_{20}+xb_{21})\\
&=&
g_1(e_1)b_{10}+g_2(e_1)b_{20}+x[g_1(e_1)b_{11}+g_2(e_1)b_{21}].
\end{eqnarray*}
The malicious node thus knows
\[
c_{10}=g_1(e_1)b_{10}+g_2(e_1)b_{20},~c_{11}=g_1(e_1)b_{11}+g_2(e_1)b_{21}.
\]
Alternatively, we can rewrite
\begin{eqnarray*}
&&g_1(e_1)A_{s_1}(x)+g_2(e_1)A_{s_2}(x)  \\
&=&g_1(e_1)(a_{00}+a_{10}s_1+a_{20}s_1^2)+g_2(e_1)(a_{00}+a_{10}s_2+a_{20}s^2)\\
&&+xg_1(e_1)(a_{01}+a_{11}s_1+a_{21}s_1^2)+xg_2(e_1)(a_{01}+a_{11}s_2+a_{21}s_2^2)\\
&=&
a_{00}(g_1(e_1)+g_2(e_1))+a_{10}(g_1(e_1)s_1+g_2(e_1)s_2)+a_{20}(g_1(e_1)s_1^2+g_2(e_1)s_2^2)\\
&&+x[a_{01}(g_1(e_1)+g_2(e_1))+a_{11}(g_1(e_1)s_1+g_2(e_1)s_2)+a_{21}(g_1(e_1)s_1^2+g_2(e_1)s_2^2)].
\end{eqnarray*}
Since the malicious node knows $g_1(e_1)+g_2(e_1)$, $g_1(e_1)s_1+g_2(e_1)s_2$
and $g_1(e_1)s_1^2+g_2(e_1)s_2^2$, and by iterating the computations for the
second incoming edge, it can form the following system of linear
equations:
\[
\left(
\begin{array}{ccc}
a_{0,0} & a_{1,0} & a_{2,0}\\
a_{0,1} & a_{1,1} & a_{2,1} \\
\end{array}
\right)
G
=
\left(
\begin{array}{cc}
c_{10} & c_{2,0}\\
c_{1,1}& c_{2,1}
\end{array}
\right)
\]
where
\[
G=
\left(
\begin{array}{cc}
g_1(e_1)+g_2(e_1)&g_1(e_2)+g_2(e_2)\\
g_1(e_1)s_1+g_2(e_1)s_2&g_1(e_2)s_1+g_2(e_2)s_2\\
g_1(e_1)s_1^2+g_2(e_1)s_2^2&g_1(e_2)s_1^2+g_2(e_2)s_2^2
\end{array}
\right).
\]
If $R_1$ is not a verifying node, it should prepare an attack based on the 
knowledge of this system of equations. We can illustrate the condition 
$H \leq M$ required for security. Suppose that it were not the case, that is 
$H=2$ but we have only $M=1$, meaning that only two polynomials $P_0$ and 
$P_1$ are used to create the authentication tag, then the matrix $G$ would be
a $2\times 2$ matrix, and thus could be very likely invertible, thus
allowing the malicious node to recover the secret coefficients of the source
private key, although the node cannot decode the message.

Now if furthermore $R_1$ has a private key $[P_0(x_1),P_1(x_1),P_2(x_1)]$,
it further knows that
 \[
(1,~x_1)
\left(
\begin{array}{ccc}
a_{0,0} & a_{1,0} & a_{2,0}\\
a_{0,1} & a_{1,1} & a_{2,1} \\
\end{array}
\right)=
(P_0(x_1),~P_1(x_1),~P_2(x_1)).
\]
Let us assume for this example that the first row of $G$ has non-zero
coefficients, so that both coefficients are invertible. We set
\begin{eqnarray*}
\gamma_1 & = & (g_1(e_1)s_1+g_2(e_1)s_2)(g_1(e_1)+g_2(e_1))^{-1}\\
\gamma_2 & = & (g_1(e_2)s_1+g_2(e_2)s_2)(g_1(e_2)+g_2(e_2))^{-1}
\end{eqnarray*}
and we can rewrite $G$ as
\[
\begin{array}{c}
\left(
\begin{array}{cc}
g_1(e_1)+g_2(e_1)&g_1(e_2)+g_2(e_2)\\
\gamma_1(g_1(e_1)+g_2(e_1))&\gamma_2(g_1(e_2)+g_2(e_2))\\
\gamma_1^2(g_1(e_1)+g_2(e_1))&\gamma_2^2(g_1(e_2)+g_2(e_2))\\
\end{array}
\right)
=
\\
\left(
\begin{array}{cc}
1 & 1\\
\gamma_1 & \gamma_2 \\
\gamma_1^2 & \gamma_2^2 \\
\end{array}
\right)
\left(
\begin{array}{cc}
g_1(e_1)+g_2(e_1) & 0 \\
0 & g_1(e_2)+g_2(e_2)
\end{array}
\right).
\end{array}
\]
It is a straightforward computation to check that the matrices
\[
rA=
r
\left(
\begin{array}{ccc}
-x_1\gamma_1\gamma_2 & x_1\gamma_1+x_1\gamma_2 & -x_1\\
\gamma_1\gamma_2 & -\gamma_1-\gamma_2 & 1 \\
\end{array}
\right),~r\in\FF_2^3
\]
satisfy the system of equations $AG={\bf 0},~XA={\bf 0}$, where $X=(1,~x_1)$.
\end{example}

%
%

\section{Multicast Goodput Analysis}
\label{sec:perf}

In this section, we discuss the performance of our scheme in terms
of multicast throughput and multicast goodput. The multicast
goodput is analyzed to assess the impact of pollution attacks in
network coding systems and to show how much the multicast
throughput is degraded under such attacks.

The analysis starts with definitions of multicast throughput and
multicast goodput. We then derive their characterizations in our
setting, depending on whether the proposed authentication scheme
is used. We provide three exemplary
topologies with various numbers of intermediate nodes, shown in
Figure \ref{pic:topo}, to illustrate the multicast throughput
gains obtained using our scheme.

\subsection{Definitions}

Recall that we have a single source $S$, sending $n$ messages to $T$ 
destination nodes $D_1,\ldots,D_T$, while $\Vc$ will denote the set of $V$ 
receivers $R_1,\ldots,R_V$ that can verify the authentication tags.
The intermediate nodes may or may not have been corrupted by malicious 
messages. We will denote by $\Rc_c$ a set of intermediate nodes with 
corrupted messages in their incoming buffers and by $\Rc_g$ a set of 
intermediate nodes with ``good" (i.e. non-corrupted) messages in their 
incoming buffers, with cardinality respectively $|\Rc_c|=r_c$ and 
$|\Rc_g|=r_g$.

We consider a single multicast session $\mathfrak{s}(S,n, \Rc, \Dc,r_c)$ where 
the source node $S$ delivers $n$ messages to all nodes in a destination set 
$\Dc \subseteq \{D_1,\ldots,D_T\}$ through multi-hop paths in a set $\Rc$ of 
intermediate nodes containing $r_c$ corrupted nodes.

We define the following performance metrics:
\begin{itemize}
\item
The message rate of a multicast session $\mathfrak{s}(S,n, \Rc, \Dc, r_c)$ is 
termed the \emph{multicast throughput}, and is denoted by $R_{S\Dc}$.
\item
The  rate of messages successfully delivered to each destination per 
session $\mathfrak{s}$ is termed \emph{throughput per destination} and is 
denoted by $R_{SD_i}$ for the destination $D_i$.
\item
The  rate of non-corrupted messages of  a multicast session 
$\mathfrak{s}(S,n, \Rc, \Dc, r_c)$ is termed the \emph{multicast goodput}. 
It is denoted by $G_{S\Dc}$ if our scheme is used, and by $G'_{S\Dc}$ 
otherwise.
\item
The rate of non-corrupted messages delivered to each destination per session 
$\mathfrak{s}$ is termed the \emph{goodput per destination} and is denoted by 
$G_{SD_i}$ for the destination $D_i$ if our scheme is used, and by $G'_{SD_i}$ 
otherwise.
\end{itemize}


\subsection{Multicast goodput analysis without the authentication scheme}

Pollution attacks degrade the multicast throughput $R_{S\Dc}$ of a session with a degradation
factor $\alpha \in [0,1]$, resulting in a multicast goodput of the
form: $$G'_{S\Dc}= \alpha R_{S\Dc}.$$

The multicast goodput of a session $\mathfrak{s}(S,n, \Rc, \Dc, r_c)$ depends on the topology
of the network and is expressed by the following expression:

\begin{equation}\label{eq:goodput0}
    G'_{S\Dc}= (1-\frac{n_{pc}}{n_{e_\Dc}})R_{S\Dc}
\end{equation}
where $n_{pc}$ is the number of paths corrupted by $r_c$, i.e., 
from the corrupted intermediate nodes $\Rc_c$ to the destinations in
$\Dc$; and $n_{e_\Dc}$ is the number of incoming edges in the
destination set $\Dc$. The multicast goodput varies depending on the
positions of the $r_c$ corrupted intermediate nodes in the
network.

The average multicast goodput of a session $\mathfrak{s}(S,n, \Rc,
\Dc, r_c)$ over all $j$ positions of the $r_c$ corrupted
intermediate nodes in the network is expressed by:
\begin{equation}\label{eq:goodput1}
    \tilde{G'}_{S\Dc}= \frac{\sum^\lambda_{j=1} G'_{S\Dc}(j)}{\lambda}
\end{equation}
where $\lambda$ is the combination of $r_c$ over $r$:
$\lambda=C^{r_c}_{r}= \frac{r!}{r_c!(r-r_c)!}.$

\subsection{Multicast goodput analysis with the authentication scheme}

With our authentication tags, if $\Rc \subset \Vc$, intermediate nodes in 
the network can then verify the integrity and origin of the messages received 
without having to decode.
They can detect and discard the corrupted messages in-transit that fail 
the verification.

The corrupted messages are discarded at their entrance in the network, and
therefore do not propagate in the network towards the destinations.
The multicast goodput is thus not degraded ($\alpha=1$), and equal to 
the multicast throughput:

\begin{equation}\label{eq:goodput3}
   G_{S\Dc}= R_{S\Dc}.
\end{equation}

The average multicast goodput gain offered by our scheme is expressed as
follows:

\begin{eqnarray}\label{eq:goodput4}
   \tilde{Gain}&=& \tilde{G_{S\Dc}} - \tilde{G'_{S\Dc}}\\
   &=&\tilde{R_{S\Dc}} - \tilde{G'_{S\Dc}}
\end{eqnarray}
where $\tilde{G'_{S\Dc}}$ is the average multicast goodput obtained without 
the use of our scheme.

Let us now present a few examples based on different topologies.

\begin{figure}
  \begin{center}
    \leavevmode
    \epsfxsize=10cm
   \epsfbox{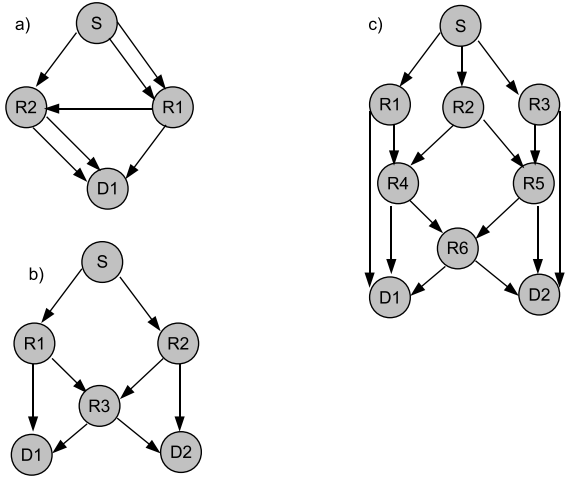}
     \caption{Examples of network topologies.}
    \label{pic:topo}
  \end{center}
\end{figure}

\noindent{\bf Topology a).} In Figure \ref{pic:topo}, we consider the 
topology a) with various configurations $\Rc_c$ (this is also the topology 
discussed in Example \ref{ex:net1}).
\begin{itemize}
\item
If $n=3$, $r_c=1$ and our scheme is not used, we obtain:
$$G'_{S\Dc}=\{\frac{1}{3}, \frac{2}{3}\} R_{S\Dc}.$$
\item
If $n=3$,  $r_c=2$ and our scheme is not used, then
$R_g=0$ and : $$G'_{S\Dc}= 0.$$
\end{itemize}

\begin{table}
\centering
\begin{tabular}{|c|c|c|c|c|c|}
  \hline
  $n$ &  $r_c$ & $min(G'_{S\Dc})$ & $max(G'_{S\Dc})$ & $\tilde{G'}_{S\Dc}$\\\hline
  3 &  1&$\frac{1}{3}R_{S\Dc}$&$\frac{2}{3}R_{S\Dc}$ &$\frac{1}{2}R_{S\Dc}$\\\hline
  3 & 2&0&0&0\\\hline
\end{tabular}
\normalsize
\caption{Multicast Goodput results for Topology a) }
\label{tab:table0}
\end{table}

\noindent{\bf Topology b).}
In Figure \ref{pic:topo}, we consider the topology b) with again various
configurations of $\Rc_c$.

If $n=2$,  $r_c=1$ and our scheme is not used, we have two possibilities 
for the intermediate receiver that holds corrupted packets:
\begin{itemize}
\item
If the intermediate receiver with corrupted messages is on the first
hop from the source (i.e., $R_1$, $R_2$), then
$$G'_{SD_i}=\{0,\frac{1}{2}\} \times R_{SD_i},$$
$$G'_{S\Dc}= \frac{1}{4}R_{S\Dc}.$$
\item
If the intermediate receiver with corrupted messages is on the second hop
from the source (i.e., $R_3$), then
$$G'_{SD_i}=\frac{1}{2}R_{SD_i},$$ 
$$G'_{S\Dc}= \frac{2}{4}R_{S\Dc}=\frac{1}{2}R_{S\Dc}.$$
\end{itemize}

If $r_c=2$ and our scheme is not used, then we have $r_g=1$, and there are 
two possibilities again:
\begin{itemize}
\item 
If the intermediate receivers with corrupted messages are on the first
hop from the source (i.e., $R_1$, $R_2$), then
$$G'_{SD_i}=0;G'_{S\Dc}= 0.$$
\item 
If one intermediate receiver with corrupted messages is on the second
hop from the source (i.e. $R_3$) and the other is on the first hop from the
source (i.e., $R_1$, $R_2$), then
$$G'_{SD_i}=\{0,\frac{1}{2}\} \times R_{SD_i},$$ 
$$G'_{S\Dc}= \frac{1}{4}R_{S\Dc}.$$
\end{itemize}

If $r_c=3$ and our scheme is not used, we have $r_g=0$
and $G'_{S\Dc}=0$. The multicast goodput results are summarized in
Table \ref{tab:table1}:
\\
\begin{table}
  \centering
\begin{tabular}{|c|c|c|c|c|c|}
  \hline
  $n$ &  $r_c$ & $min(G'_{S\Dc})$ & $max(G'_{S\Dc})$& $\tilde{G'}_{S\Dc}$\\\hline
  2 &  1&$\frac{1}{4}R_{S\Dc}$&$\frac{1}{2}R_{S\Dc}$&$\frac{1}{3}R_{S\Dc}$\\\hline
  2 & 2&0&$\frac{1}{4}R_{S\Dc}$&$\frac{1}{6}R_{S\Dc}$\\\hline
  2 & 3&0&0&0\\\hline
\end{tabular}
  \normalsize
  \caption{Multicast Goodput results for Topology b) }\label{tab:table1}
\end{table}

\noindent{\bf Topology c).} In the topology c), we consider also various 
configurations of  $\Rc_c$. The multicast goodput results 
are summarized in Table \ref{tab:table2}.

\begin{table}
  \centering
  \normalsize
\begin{tabular}{|c|c|c|c|c|c|}
  \hline
  $n$ & $r_c$ & $min(G'_{S\Dc})$ & $max(G'_{S\Dc})$& $\tilde{G'}_{S\Dc}$\\\hline
  3 & 1&$\frac{2}{6}R_{S\Dc}$&$\frac{4}{6}R_{S\Dc}$&$\frac{4}{9}R_{S\Dc}$\\\hline
  3 & 2&0&$\frac{3}{6}R_{S\Dc}$&$\frac{4}{15}R_{S\Dc}$\\\hline
  3 & 3&0&$\frac{2}{6}R_{S\Dc}$&$\frac{11}{60}R_{S\Dc}$\\\hline
  3 & 4&0&$\frac{2}{6}R_{S\Dc}$&$\frac{1}{9}R_{S\Dc}$\\\hline
  3 & 5&0& $\frac{1}{6}R_{S\Dc}$&$\frac{1}{18}R_{S\Dc}$\\\hline
  3 & 6&0&0&0\\
  \hline
\end{tabular}
    \caption{Multicast Goodput results for Topology c) }\label{tab:table2}
\end{table}

If we now consider the goodput gains with our scheme for topologies 
a), b), c), we get that for all $r_c$,  $G_{S\Dc}= R_{S\Dc}$. In the three
topologies, our scheme offers multicast goodput gains that are given in
Table \ref{tab:table4}. As the number of corrupted messages injected increases in the network, the
average multicast goodput gain naturally tends towards $1$.
\begin{table}
  \centering
\begin{tabular}{|c|c|c|c|c|}
  \hline
   $r_c$& Topology a) & Topology b) & Topology c)\\\hline
  $1$ & 0,5&0,66&0,55\\\hline
  $2$&1 & $0,83$&0,73\\\hline
  $3$&- & 1&0,81\\\hline
  $4$&-&- &0,88\\\hline
  $5$&-&- &0,94\\\hline
  $6$&-&- &1\\\hline
\end{tabular}
  \normalsize
  \caption{Average Goodput Gains obtained with our scheme}\label{tab:table4}
\end{table}

\section{Application to File distribution}
\label{sec:applic}

In this section, we present how the proposed $(k,V,M)$ authentication scheme 
could be easily applied to content or file distribution.
For content distribution over an IP-based network with our scheme, at
most $M$ messages forming the file to be distributed can be transmitted by
the source through the network in an authenticated way using the same key.
For our scheme to be secure against a coalition of $k-1$ receivers, we recall 
the following rules:
\begin{itemize}
\item 
$M\geq n$, where $n$ is the number of messages to be sent by the source.
\item 
$M\geq H$, where $H$ is the maximum number of incoming edges in a coalition 
of malicious nodes.
\end{itemize}

We also define $N$ as the size of the generation of IP packets carrying one 
message authenticated by one tag. Figure \ref{pic:exple} illustrates the 
relation between an IP packet and a message.

\begin{figure}
  \begin{center}
    \leavevmode
    \epsfxsize=10cm
   \epsfbox{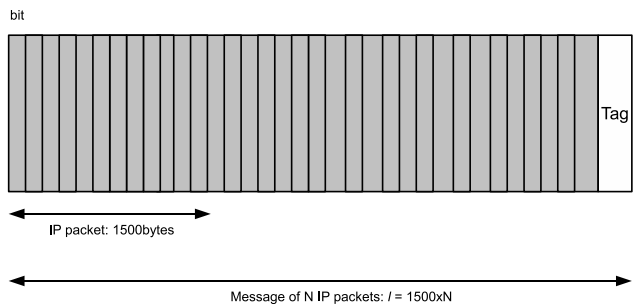}
     \caption{Structure of a message}
    \label{pic:exple}
  \end{center}
\end{figure}

In a practical scenario, the following should be considered:
\begin{itemize}
\item
a message consists of $l$ symbols $s_{ij}$ with a symbol being bit.
\item
one message authenticated by one tag consists of $N$ IP packets (also called
a generation).
\item
IP packets are 1500 bytes long (12000 bits) with a payload of 1480 bytes.
\item
The message length $l$ can be expressed in bits and in bytes. We refer to
$l_{bits}$ to the message length expressed in bits and to $l_{bytes}$ to the
message length expressed in bytes. $$l_{bytes}= 1500 \times N $$ $$l_{bits}=8
\times l_{bytes} = 12000\times N.$$
\end{itemize}

For $M\leq l_{bits}$, we have $M \leq 12000\times N$, which means that the
source can use the same key to tag at most $M=12000\times N$ messages of
length $12000N$ bits (carried over $N$ IP packets that are 12000 bits long).

Destinations can download a file with at most the following size in bytes
(including headers): $M\times l_{bytes}= l_{bits}\times l_{bytes}=
8\times l_{bytes}^2=8\times (1500N)^{2} = 18 \times 10^6 \times N^2$ bytes.
The destinations can  use the same key to authenticate a file download of at
most $18N^2$ MBytes when one tagged message is carried over $N$ IP packets.

A receiver node can have at most $12000\times N$ incoming edges and the source $S$ can send $n \leq 12000 \times N $ messages.

The scenarios in Table \ref{tab:table5} show what should be the size of an IP packet
generation to allow the  distribution of a given file to be authenticated
under the same key.

For distributing a file that is 18MBytes, it is sufficient for the source to
send one tagged message in one IP packet of 1500bytes . The source sends then
12 000 messages tagged that form the 18MBytes file. Any destination can verify
with the same key each tag attached to the 12 000 messages.

For distributing a file that is 1.8GBytes, the source generates tagged
messages of size 15KBytes.  Each message is sent in a generation of 10 IP
packets. The source sends 120K messages tagged that form the file. At the
destination, the same key can be used to verify the tags of the 120K messages
received.
\begin{table}
  \centering
\begin{center}
\begin{tabular}{|c|c|c|c|}
\hline
File size & Generation Size& Message length & Nb of messages authenticated  \\
(bytes)& N &l (bytes)& by the same key M \\\hline
18M& 1 & 1500 & 12 000\\\hline
72M&2&3000& 24 000\\\hline
1.8G& 10 & 15K & 120 000 \\\hline
4.05G&15&22.5K&180 000\\\hline
\end{tabular}
\end{center}
  \normalsize
  \caption{Parameters of our scheme for distribution of files of variable sizes}
  \label{tab:table5}
\end{table}

%
%

\section{Conclusion}
In this paper, we have proposed an unconditionally secure authentication 
scheme that provides multicast linear network coding with message integrity 
protection and source authentication. The resulting scheme offers robustness 
against pollution attacks from outsiders and from $k-1$ insiders. Our solution 
allows the source to generate authentication tags for up to $M$ messages with 
the same key and the intermediate nodes to verify the authentication tags of 
the packets received and thus to detect and discard the malicious packets that
fail the verification. The performance analysis showed that our scheme offers 
goodput gains that tend towards $1$ with increasing corrupted packets in the 
network. Our scheme can be used to authenticate with the same key a file 
download of at most $18N^2$ MBytes when one tagged message is carried over $N$ 
IP packets.

Future work will involve optimization of the parameters involved in the
authentication scheme for a more efficient solution.
Another aspect to consider in the future is to offer more flexibility over 
the sender as the scheme proposed here requires the sender to be designated.

%
%

\section*{Acknowledgment}

The work of Fr\'ed\'erique Oggier is supported in part by the Singapore National
Research Foundation under Research Grant NRF-RF2009-07 and 
NRF-CRP2-2007-03, and in part by the Nanyang Technological University under 
Research Grant M58110049 and M58110070.

Hanane Fathi would like to acknowledge the support of the National Institute
of Advanced Industrial Science and Technology, Japan.

Most of the ideas of this work were discussed while both authors were
visiting the AIST Research Center for Information Security, Tokyo, Japan.

%
%

\end{document}